\documentclass[10pt, twocolumn]{IEEEtran}
\usepackage{graphicx,multirow}
\usepackage{epsfig}
\usepackage{cite}
\usepackage{amssymb}
\usepackage{amsthm}
\usepackage{amsmath}
\usepackage{color}
\usepackage{bbm}
\usepackage{cite}
\usepackage{epstopdf}
\usepackage[table,xcdraw]{xcolor}
\newtheorem{theorem}{\textbf{\text{Theorem}}}

\newtheorem{approximation}{Approximation}
\newtheorem{remark}{Remark}
\usepackage{subfigure}
\usepackage{collcell, datatool}
\usepackage{pifont}

\definecolor{lightgray}{gray}{0.9}
\usepackage{balance}
\usepackage{ifmtarg}
\usepackage[labelsep=none]{caption}
\usepackage{footnote}
\usepackage{cite}
\usepackage{stfloats}

\usepackage{mathrsfs}

\usepackage{amsmath}
\usepackage{algorithm}
\usepackage[noend]{algpseudocode}

\makeatletter
\def\BState{\State\hskip-\ALG@thistlm}
\makeatother

\begin{document}

\title{Self-Organized Scheduling Request for Uplink 5G Networks: A D2D Clustering Approach}
\author{
	\IEEEauthorblockN{\large  Mohammad Gharbieh, Ahmed Bader, Hesham ElSawy, Hong-Chuan Yang, \\ Mohamed-Slim Alouini, and Abdulkareem Adinoyi
	 \thanks{ M. Gharbieh and  H.-C. Yang are with the Department of Electrical and Computer Engineering, University of Victoria, Victoria, BC V8P 5C2, Canada; e-mail: \{mohammadgharbieh, hy\}@uvic.ca. \;}
	 \thanks{A. Bader is with Insyab Wireless Limited, 1961 Dubai- UAE; e-mail: ahmed@insyab.com. \; } 
	 \thanks{  H. ElSawy was with the Computer, Electrical, and Mathematical Sciences and Engineering (CEMSE) Division, King Abdullah University of Science and Technology (KAUST), Thuwal 23955, Saudi Arabia.	He is now with the Electrical Engineering Department, King Fahd University of Petroleum and Minerals, Dhahran 31261, Saudi Arabia; e-mail: hesham.elsawy@kfupm.edu.sa.\;}	
	 \thanks{ M.-S. Alouini is with the Computer, Electrical, and Mathematical Sciences and Engineering (CEMSE) Division, King Abdullah University of Science and Technology (KAUST), Thuwal 23955, Saudi Arabia; e-mail: slim.alouini@kaust.edu.sa.  \;}	
 	 \thanks{A. Adinoyi is with Saudi Telecom Company (STC); e-mail: aadinoyi.c@stc.com.sa \; }
	 \thanks{The work of the KAUST team was supported in part by STC under grant RGC/3/2374-01-01.\; }
}
}

\maketitle
\thispagestyle{plain}
\pagestyle{plain}

\begin{abstract}

In one of the several manifestations, the future cellular networks are required to accommodate a \emph{massive} number of devices; several orders of magnitude compared to today's networks. At the same time, the future cellular networks will have to fulfill stringent latency constraints. To that end, one problem that is posed as a potential showstopper is extreme congestion for requesting uplink scheduling over the physical random access channel (PRACH). Indeed, such congestion drags along scheduling delay problems. In this paper, the use of self-organized device-to-device (D2D) clustering is advocated for mitigating PRACH congestion.  {To this end, the paper proposes two D2D clustering schemes, namely; Random-Based Clustering (RBC) and Channel-Gain-Based Clustering (CGBC). Accordingly, this paper sheds light on random access within the proposed D2D clustering schemes and presents a case study based on a stochastic geometry framework. For the sake of objective evaluation, the D2D clustering is benchmarked by the conventional scheduling request procedure. Accordingly, the paper offers insights into useful scenarios that minimize the scheduling delay for each clustering scheme. Finally, the paper discusses the implementation algorithm and some potential implementation issues and remedies}. 

\end{abstract}

\vspace{-2mm}
\section{Introduction}
\vspace{-1mm}
\IEEEPARstart{T}{he} next generation of cellular networks is expected to involve a massive number of connected devices varying from sensors, smart objects, machines, all the way to smartphones and vehicles \cite{what.will.5g.be}.  For future networks to enable a broad spectrum of new usage and applications, the cellular infrastructure must support a mixture of human-type and machine-type communications with ever-increasing traffic levels. In fact, 5G networks are expected to handle a 1000-fold increase in capacity \cite{7390494}, an appreciable portion of which is uplink traffic \cite{first.mile.KAUST}. Within this context, a primary challenge pertains to the uplink scheduling request that is performed via random access ({\rm RA}) procedure over the physical RA channel (PRACH). Particularly, devices with uplink traffic need to go through {\rm RA} procedures over the PRACH to request resource allocation from the base station ({\rm BS})~\cite{sesia2009lte}. As the number of devices grows, contention over scarce PRACH resources escalates substantially thus leading to a large number of devices dropping off the {\rm RA} process, and high volume of unserved traffic demand as discussed in \cite{first.mile.KAUST} and illustrated through the experimental data shown in Fig.~\ref{fig:traffic}. The figure shows that while there are enough resources to schedule more uplink traffic, such resources are wasted because the devices fail to pass their scheduling request to the BS through the PRACH. Hence, it is clear that {\rm RA} scheduling requests lead to congestion that needs to be alleviated to fulfill the foreseen 5G performance.

\begin{figure}[t!]
\begin{center}
\includegraphics[width=3.4 in]{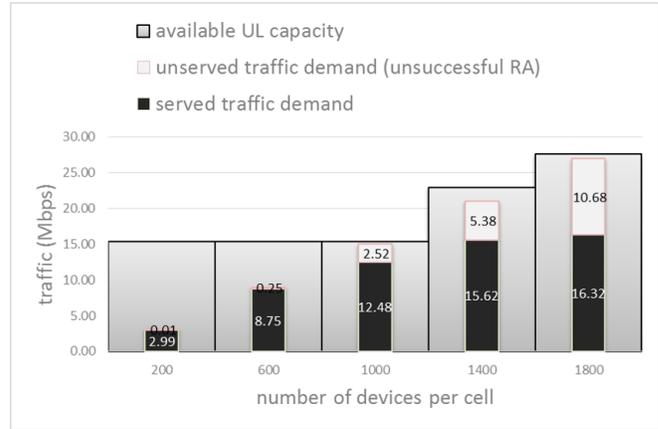}	
\end{center}
\vspace{-3mm}
\caption{: At low to moderate device counts, minimal to negligible levels of collisions occur on the {\rm RA} radio resources. However, mobile networks are already past this phase, i.e. the number of devices in a small cell is rapidly growing beyond 1000 \cite{first.mile.KAUST}.}
\label{fig:traffic}
\end{figure}

The most straightforward proposition to alleviate {\rm RA}  congestion is to simply allocate additional radio resources for PRACH. This option obviously reduces the available resources for scheduling uplink data traffic. Moreover, allocating spectrally adjacent blocks for {\rm RA}  increases computational complexity at the {\rm BS} side due to parallelized processing~\cite{sesia2009lte}. As such, it is not an appealing solution for the vendors. Another obvious proposition is to densify {\rm BS} deployments as a mean to reduce congestion. Nonetheless, densification makes sense to mobile network operators only up to a certain limit. Beyond that, it ceases to offer either economic benefit \cite{densification.economics} or performance improvement \cite{Ultra_Dense_AlAmmouri}. From the economic perspective, right-of-way and site acquisition costs may become major challenges. From the performance perspective, there is a critical density after which the coverage probability and rate degrade with BS density due to the overwhelming inter-cell interference. Another drawback for network densification is the increased handover rate for mobile users which consumes physical resources and incurs a delay~\cite{Rabe1,Rabe2}.

To this end, a distributed self-organized RA procedure is better positioned to accommodate this tremendous uplink demand in the future cellular networks. Indeed,  standards for Long Term Evolution (LTE) have identified self-organization as a vital requirement for future networks \cite{6157579}. The self-organized random access can response to actual network variations in near-real-time. Moreover, the self-organization random access is less costly since it entails the use of significantly less resource allocation complexity and administrative overhead.  

\begin{figure}[t!]
	\begin{center} 
    \vspace{-7mm}
	\hspace*{-30mm}
		\includegraphics[width=5.5in]{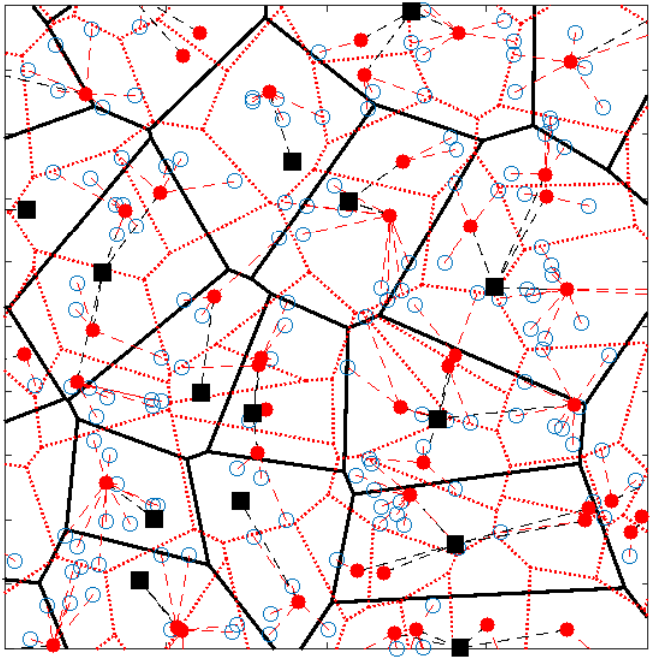}	
	\end{center}
	\vspace{-12mm}
	\caption{: Network realization for the system model for Devices-to-BS ratio $\alpha=16$ and cluster head selection probability {$\delta=0.25$}. The {\rm BS}s are denoted by black squares, {\rm CH}s are denoted by the red circles, and the {\rm CM}s are denoted by blue circles. The Voronoi cell of the {\rm BS}s are denoted by the solid black lines while the Voronoi cell of the {\rm CH}s are denoted by the dotted red lines. The black dashed lines denote the associations of the {\rm CH}s to the BSs while the red dashed line denote the associations of the {\rm CM}s to the {\rm CH}s. In this model, uplink and downlink interfaces can be decoupled, i.e., downlink traffic and signaling can be transmitted directly by the {\rm BS} to the {\rm CM}s.}
	\label{fig:scheme}
\end{figure}
\vspace{-2mm}
\subsection{Prior Work}
\vspace{-1mm}
{ Device-to-device (D2D) relaying has been classically exploited within LTE networks, i.e., in-band D2D, mainly as a coverage improvement solution \cite{D2D.survey}. A corollary to coverage enhancements is indeed a boost in throughput or the spectrum efficiency through traffic offloading from cellular networks. In the stochastic geometry literature, different network architectures and systems were proposed to study and assess the spectrum sharing of in-band D2D communication \cite{D2D_jeff,6928445,7073589,Gjess_D2D,7548294,d2d_pca}.  The authors in \cite{d2d_Harpreet} analyze out-band D2D for uniform and k-closest  content availability in terms of the coverage probability and the area spectral efficiency. Moreover, \cite{d2d_economic} studies the economic aspect of downlink traffic offloading via D2D for in-band and out-band operating modes. Furthermore, \cite{d2d_Haenggi} develops an approach to model single- and multi-cluster wireless networks and study the coverage probability for closest-selection and uniform-selection strategies. It is worth to highlight that the model in \cite{d2d_Haenggi} is suitable for downlink cellular networks or ad hoc networks. However, the idea of aggregating the uplink generated traffic within a cluster using out-band D2D communication in the presence of  the cellular networks was investigated in \cite{Asadi:2013:CIO:2507924.2507929,7248779,6686520}.} The authors of \cite{Asadi:2013:CIO:2507924.2507929} provide analytical expressions for the throughput and power consumption for a point-to-point scenario. In \cite{7248779} the authors studied the latency-power trade-off of aggregating traffic on D2D links, where they showed  that the transmit power can be reduced but at the expense of higher latency. Furthermore,  \cite{6686520} defines a protocol stack for the D2D communication in cellular networks, and uses system-level simulation to show the throughput improvement. However, none of the proposed protocols state the criterion and/or the effective scenarios to activate the D2D communication \cite{D2D.survey}. Most notably, D2D relaying literature rarely touches upon its advents relaying scheduling request to relieve {\rm RA} congestion over the PRACH. While it may be quite intuitive, but a proper quantification of such an advantage is still missing out from literature.

\vspace{-2mm}
\subsection{Contributions}
\vspace{-1mm}
{This paper\footnote{This work is presented in part in \cite{8254606}.} presents an out-band D2D relaying setup (e.g., WiFi Direct~\cite{6686520}) that can be exploited to boost the {\rm RA} performance and LTE network capacity as well in dense networks. Within the context of this paper, the D2D paradigm refers to the situation where a number of devices cluster themselves together through an out-band link and assign a cluster head (CH)  as depicted in Fig.~\ref{fig:scheme}.  Uplink scheduling requests from cluster members (CM)s are forwarded to the assigned CH over unlicensed spectrum. The CH aggregates the requests from the CMs into larger ensembles and transmits one {\rm RA}  request per ensemble over the LTE interface. The BS process the {\rm RA} request from the {\rm CH}s and sends the uplink resources scheduling to the {\rm CM}s directly through the downlink signaling.  Without any doubt,  such clustering  relaxes the congestion over the LTE {\rm RA} resources since the number of LTE {\rm RA}  requests is reduced, and hence, reduces the latency for resource allocation over the LTE interface. The problem is not trivial though. One has to consider whether the average access delay perceived by a device is actually enhanced by virtue of D2D clustering or not. This has to be evaluated in light of intra- and inter-cell interference. Indeed, this is the crux of the study carried out in this paper. The contribution of this work can be summarized as follows:
\begin{enumerate}
	\item The paper proposes a self-organized D2D clustering in which each CH acts as a virtual Access Point ({\rm AP}) over an LTE connection to boost the {\rm RA} performance.
	\item The paper considers two CH assignment mechanisms:
		 \begin{itemize}
		  	\item Random-Based Clustering (RBC) in which each device is assigned randomly with probability ($\delta$) to be a CH.
		  	\item Channel-Gain-Based Clustering (CGBC) in which only the devices with channel gain greater than a threshold ($\tau$) are assigned to be CHs.
		 \end{itemize}
	\item We present analytical expressions, based on stochastic geometry which takes into account the spatial intra and inter-cell/cluster interference sources to assess the transmissions success probabilities. Consequently, we quantify the average uplink scheduling delay for the D2D clustering.
	\item The proposed self-organized D2D clustering scheme is benchmarked by the conventional {\rm RA}  procedure where all the devices have to send {\rm RA} request to the BS over the LTE interface.
	\item We quantify the critical device density beyond which the self-organized D2D clustering, for relaying uplink scheduling requests,  offers performance gains (i.e., reduction in channel access delay).
	\item For the range of device densities where D2D is feasible we answer the following crucial question: How to activate the D2D relaying and how should the {\rm CH} be assigned? i.e., what are the suitable design parameters for these setups?
\end{enumerate}
 The results show that for the RBC and low device intensity is actually better to follow the conventional {\rm RA}  procedure. However, the self-organized D2D relaying scheme starts to pay off as the intensity grows. On the other hand, when the channel gains are considered in the CGBC scheme, D2D clustering provides higher delay reduction. Moreover, there is an optimal {\rm CH} selection probability ($\delta$) or a channel gain threshold ($\tau$) that minimizes the average delay for every device intensity.}
\vspace{-2mm}
\subsection{Notation \& Organization}
\vspace{-1mm}
Throughout the paper, we use the math italic font for scalars, e.g., $x$. We use the calligraphic font, e.g.,  $\mathcal{X}$ to represent a random variable (RV) while the math typewriter font, e.g., $\mathtt{x}$ is used to represent its instantiation. Moreover, $\mathbb{E}_ \mathcal{X}\{\cdot\}$, $F_{\mathcal{X}}\left(\cdot \right) $,   $\bar{F}_{\mathcal{X}}\left(\cdot \right)$, and $\mathscr{L}_\mathcal{X}\left(\cdot \right)$ denote, respectively, the expectation, the cumulative distribution function (CDF), the complementary cumulative distribution function (CCDF), and the Laplace Transform (LT) of the PDF of the random variable $\mathcal{X}$. We use $\mathbb{P} \{\cdot\}  $ to denote the probability. $\mathbbm{1}_{\{ \cdot \}}$ is the indicator function which has value of one if the statement $\{ \cdot \}$ is true and zero otherwise. $\Gamma(\cdot)$ indicates the Gamma function and ${}_2 F_1(.)$ is the Gaussian hypergeometric function. The imaginary unit is denoted by $j=\sqrt{-1}$ and imaginary component of a complex number is denoted as $\text{Im}\{.\}$. {Lastly, $y^*$ denotes the optimal value of $y$.}

The rest of the paper is organized as follows.  Section \ref{sec:Protocol} points out a high-level protocol description for the D2D clustering scheme. Section \ref{sec:System_Model} models the physical layer attributes of the  communication system and highlights the performance metrics. Section \ref{sic:Performance_Analysis} characterizes the D2D clustering protocols, while Section \ref{Results} provides the numerical results and insights. Section \ref{sec:Imp} sheds light on some implementation obstacles and pertinent remedies and recommendations. Finally, Section \ref{sec:Conclusions} summarizes and concludes the paper.

\vspace{-2mm}
\section{Overiew of the Protocol}\label{sec:Protocol}
\vspace{-1mm}
The goal of this section is not to define an exhaustive protocol stack for self-organized  D2D clustering within LTE networks as in \cite{6686520}. Rather, the aim is to point out a high-level description that helps to digest the presented scheme. The self-organized D2D clustering process can be summarized as follows:
\begin{enumerate}
	\item {\bf D2D Clustering Initiation Order}: the {\rm BS} broadcasting this order along with the chosen value of the {\rm CH} selection probability ($\delta$) over the downlink signaling.
	\item {\bf{\rm \bf CH} Selection}: for a given target fraction (i.e., $\delta$) of devices that are required to act as {\rm CH}s, the clustering for each scheme is performed as follows:
	
	\begin{itemize}
		\item for the RBC scheme, each device has a probability $\delta$ to be a {\rm CH}. The selection can be done in a distributed manner, i.e., without the control of the {\rm BS}. This can be done via a generating a random number $\mathtt{v} \in [0,1]$, and hence, the device becomes a {\rm CH} if $\mathtt{v} \leq \delta$.
		
		\item for the CGBC scheme, each device has to estimate its channel over the LTE interface and the device becomes a {\rm CH} if the channel gain is greater than $\tau$, where $\tau$ is identified such that the fraction $\delta$ of devices are selected as CHs. Let $\mathtt{h}$ be the channel gain, then $\tau$ is identified through the inverse of the CCDF of $\mathtt{h}$ as $\tau = \bar{F}_{\mathtt{h}}^{-1}(\delta)=-\ln(\delta)$. Therefore, the selection can be done in a distributed manner as well. 
	\end{itemize}
	
	\item {\bf {\rm \bf CH} Announcement}: Once the {\rm CH}s are identified each {\rm CH} selects a frequency channel in the D2D spectrum and broadcasts its D2D-Identification (D2D-ID) to declare itself as a {\rm CH}.
	\item {\bf Cluster Formation}: each {\rm CM} scans the D2D channels searching for CHs broadcast messages. The {\rm CM}s associate themselves to their nearest {\rm CH}, by measuring the received signal strength (RSS) and selecting  {\rm CH} with the highest RSS. Through the association phase, the {\rm CM}s send their D2D-ID along with their LTE-ID, e.g., SAE-Temporary Mobile Subscriber Identity (S-TMSI).
	\item {\bf{ Cluster Registration at {\rm \bf BS}}}: each {\rm CH} sends the D2D-LTE-ID association table to its serving {\rm BS} via uplink LTE channel. Such a table is vital for the uplink resources scheduling so that the {\rm BS} transmits the downlink signaling directly to the {\rm CM}s.
\end{enumerate}

{Upon the D2D cluster formation, the {\rm CM}s relay their uplink scheduling requests via the {\rm CH}. The {\rm CH} in turn, stamps the requests by the D2D-ID of the {\rm CM}s, then aggregates all the {\rm RA} requests from the {\rm CM}s with its own request into a larger ensemble, and transmits one scheduling request via {\rm RA} on the shared LTE PRACH for each ensemble. }

\vspace{-2mm}
\section{System Modeling \& Assumptions}\label{sec:System_Model}
\vspace{-1mm}
After describing the clustering process from a protocol point of view, this section portrays the modeling attributes of the proposed self-organized clustering schemes from a physical layer point of view.  

\vspace{-2mm}
\subsection{Spatial \& Physical Layer Parameters}
\vspace{-1mm}
A two-tier cellular network is considered, namely the out-band D2D and LTE networks. Due to the disjoint spectrum allocation, the interference interactions on each network are decoupled. Since the cellular networks topologies from one location to another tend to be random, stochastic geometry is utilized to model the spatial distribution of the BSs as a point processes \cite{7733098,6620915,andrews2011tractable}. In this regards, the Poisson point process (PPP) is widely accepted and utilized due to its simplicity and practical relevance~\cite{andrews2011tractable,load_aware_jeffrey,6566864,6620915}.\footnote{ { Note that logical clustering does not change the physical locations of the devices. Hence, the PPP distribution of the devices is preserved.}} Therefore, we assume that the  BSs and the devices are spatially distributed according to two independent homogeneous PPPs  with densities $\lambda$ and $\mu$, respectively. A power-law path-loss model is considered where the signal power decays at a rate of  $r^{-\eta}$ with the propagation distance $r$, where $\eta>2$ is the path-loss exponent. In addition to the path-loss attenuation, Rayleigh block fading  is assumed within a multi-path environment, in which all the channel power gains ($\mathtt{h}$) are assumed to be independent of each other and are identically and exponentially distributed with unity power gain.\footnote{ {Note that in the case of physically clustered devices, fading is correlated.  }}

\vspace{-2mm}
\subsection{Random Access over the LTE Network}
\vspace{-1mm}
Each {\rm CH} should go through the {\rm RA}  process over the PRACH to request uplink channel access from the {\rm BS}~\cite{sesia2009lte}. The {\rm RA} process is uncoordinated and all devices can mutually  interfere with one another, which may lead to intra-cell interference in addition to the inter-cell interference. Only the {\rm CH}s are eligible to request uplink resources over the LTE interface (i.e., PRACH) from the nearest {\rm BS}. To request an uplink channel access, each  {\rm CH} randomly and independently transmits its {\rm RA} request on one of the available prime-length orthogonal Zadoff-Chu (ZC) codes defined by the LTE  PRACH preamble \cite{sesia2009lte}.

During the  {\rm RA}, each {\rm CH} uses full path-loss inversion power control with target power level $\rho_{\text{L}}$  \cite{sesia2009lte}.  Therefore, the {\rm RA} transmit power is expressed by $\mathcal{P}_{\text{{\rm RA}}}=  \rho_{\text{L}} \mathcal{R}_\text{L}^{\eta}$, where $\mathcal{R}_\text{L}$ is the distance between the {\rm CH} and its geographically closest {\rm BS}. That is, the \rm CH controls its transmit power such that the average signal power received at its serving {\rm BS} is equal to $\rho_{\text{L}}$. The target power level $\rho_{\text{L}}$ is assumed to be conveyed on downlink signaling channels by the {\rm BS}. {It is assumed that the BSs are dense enough such that each of the {\rm CH}s can invert its path-loss towards the closest BS almost surely, and hence the maximum transmit power of the IoT devices is not a binding constraint for packet transmission. Extension to fractional power control and/or adding a maximum power constraint can be done by following the methodologies in \cite{uplink2_jeff} and \cite{uplink_alamouri}.} An {\rm RA} transmission is assumed to be decodable if the signal to interference and noise ratio (SINR), denoted by $ \Upsilon_{\text{{\rm RA}}}$, is greater than a certain threshold $\theta_{\text{{\rm RA}}}$.

\vspace{-2mm}
\subsection{D2D Clustering Over an Unlicensed Spectrum}
\vspace{-1mm}
The clustering process is initiated by the {\rm BS} where the clustering criterion, i.e., RBC or CGBC, is dictated by the {\rm BS}. Each {\rm CM} associates with its nearest {\rm CH} through a single-hop link and relays the uplink scheduling requests via that link as shown in Fig.~\ref{fig:scheme}. {\rm CM}s are assumed to employ full path-loss power control with target power level $\rho_{\text{C}}$. Therefore, the transmit power is given by $\mathcal{P}_{\text{D}}= \rho_{\text{C}}\mathcal{R}_{\text{C}}^{\eta}$, where $\mathcal{R}_{\text{C}}$ is the distance between the {\rm CM} and its geographically closest CH. The target power level $\rho_{\text{c}}$ is assumed to be conveyed along with the D2D-ID of the {\rm CH} in step 3 in Section~\ref{sec:Protocol}.

Each {\rm CH} randomly and independently selects one of the $k$  available channels dedicated for D2D communications within the unlicensed spectrum. Moreover, transmission from CMs over the D2D interface is assumed to be managed by the CH via a time division multiple access (TDMA) schedule. Hence, intra-cluster interference is prohibited and only inter-cluster interference exists.  Obviously, this comes at the expense of access delay that grows with the size of the cluster. For a correct transmission at the D2D link, an SINR capture model is adopted such that a transmission can be decoded if the SINR, denoted by $\Upsilon_\text{C}$, is greater than a certain threshold $\theta_{\text{C}}$.

\vspace{-2mm}
\section{Performance Analysis} \label{sic:Performance_Analysis}
\vspace{-1mm}
Latency, or channel access delay to be more precise, is very important for several 5G application (e.g., tactile internet \cite{tactile,low_latency}). It also has been a key aspect in the design objectives of cellular systems. Therefore, the channel access (i.e., resource allocation) delay is the primary metric used in this study to evaluate the gain of the self-organized D2D clustering scheme in reducing congestion over {\rm RA} resources. Before delving into the analysis, we state the following important {approximations} that will be utilized in this paper. 

{\begin{approximation} \label{app1}
		The spatial correlations between proximate devices, in terms of transmission power, can be ignored.
	\end{approximation}
	\begin{remark} \label{rem1}
		It is well known that the sizes of adjacent Voronoi cells are correlated. Such correlation affects the number of devices, as well as, the service distance realizations in adjacent Voronoi cells. Consequently, the transmission powers at adjacent cells are correlated. Accounting for such spatial correlation would impede the model tractability. Hence, we follow the common approach in the literature and ignore such spatial correlations when characterizing the aggregate interference~\cite{uplink_alamouri,Meta_Haenggi_Control, Meta_Elsawy, Gharbieh_tcom,elsawy2014stochastic, marco_uplink, uplink2_jeff, 6516885}. However, all spatial correlations are intrinsically accounted for in the Monte Carlo simulations that are used to validate our model in Section~\ref{Results}.
	\end{remark}
	
	{\begin{approximation} \label{app2}
			For the D2D transmission, the point processes of inter-cluster interfering CMs seen at the test CH is {modeled by} a non-homogenous PPP.
	\end{approximation}}
	
	\begin{remark} \label{rem2}
		Despite that a PPP is used to model the complete set of CMs, the subset of scheduled CMs for the TDMA transmission is not a PPP. The constraint of scheduling one CM per Voronoi cell of the cluster leads to a Voronoi-perturbed point process for the set of mutually interfering CMs. Approximation~\ref{app2} is commonly used in the literature to maintain tractability~\cite{uplink_alamouri, Meta_Elsawy, elsawy2014stochastic, marco_uplink, uplink2_jeff, Gharbieh_tcom, 6516885}.
	\end{remark}
	
	{\begin{approximation} \label{app3}
			The transmission success probabilities of all  devices in the network are assumed to have a negligible temporal correlation. 
	\end{approximation}}
	
	\begin{remark} \label{rem3}
		The full path loss inversion makes the received signal power at the serving BSs/CHs independent from the service distance (i.e., the distance between the device and the serving BS/CH). Hence, the different realizations of the service distance across the devices do not affect the SINR. Furthermore, the random channel selection randomizes the set of interfering devices over different time slots, which decorrelate the interference across time. Hence, all devices in the network tend to have a negligible temporal correlation for the transmission success probabilities as shown in~\cite{Meta_Haenggi_Control, Meta_Elsawy, Gharbieh_tcom}.
\end{remark}}

It is worth mentioning that Approximations 1-3 are mandatory for tractability, regularly used in the literature, and are validated in Section~\ref{Results} via independent Monte-Carlo simulations. Based on these approximations, the  D2D cluster size, the  {\rm RA} success probability, D2D success probability, and the channel access delay are presented in, respectively, Section~\ref{sec:PDF}, Section~\ref{Sec:RA_success}, Section~\ref{Sec:D2D}, and Section~\ref{Sec:delay}. {For a quick reference, the notation used in this paper is summarized in Table~\ref{table_notation}.}

\begin{table}[t!]
	\centering
	\footnotesize
	\renewcommand{\arraystretch}{1.3}
	\begin{tabular}{||p{1cm}|p{6.5cm}|} 
		\hline  \textbf{Notation} & \textbf{Description} \\
		\hline $\mu$; $\lambda$; $\alpha$ & device density; {\rm BS} density; devices-to-BS ratio  \\
		\hline $\delta$; $\tau$ & {\rm CH} selection probability; channel gain threshold for CGBC   \\
		\hline $\theta_{\text{{\rm RA}}} $ & Detection threshold for successful {\rm {\rm RA}}  \\
		\hline $\theta_{\text{C}} $ &  Detection threshold for successful D2D transmission  \\
		\hline  $\rho_{\text{L}}$; $\rho_{\text{c}}$  &  Power control parameter for {\rm {\rm RA}}; Power control parameter for D2D   \\
		\hline $P_{\text{{\rm RA}}}$; $P_{\text{C}}$   & RA success probability; D2D transmission success probability \\
		\hline ${n_{Z}}$ & number of ZC codes dectitated for random access \\	
		\hline $k$ &  number of frequencies available for D2D transmission \\
		\hline $\eta$; $\sigma^2$ & path-loss exponent; noise power  \\		
		\hline
	\end{tabular}
	\caption{{: Summary of Notation}}\label{table_notation}
\end{table} 

\vspace{-2mm}
\subsection{D2D Cluster Size} \label{sec:PDF}
\vspace{-1mm}
Since the adopted clustering mechanisms are independent among all the devices, and by exploiting the independent thinning property of the PPP  \cite{martin_book}, the {\rm CH}s constitute a PPP with intensity $\delta \mu$ \cite{martin_book}. Similarly, the {\rm CM}s constitute a PPP with intensity $(1-\delta) \mu$. Moreover, due to the nearest {\rm CH} association, the footprint of each {\rm CH} can be expressed by a Voronoi cell with size and shape depending on the locations of its neighboring {\rm CH}s as depicted in Fig.~\ref{fig:scheme}. Therefore, the number of {\rm CM}s associated to each {\rm CH} is random. Let $\mathcal{N}$ denote the number of {\rm CM}s served by a generic {\rm CH}, following  \cite{6576413}, the probability mass function of  $\mathcal{N}$ is given by:
\setcounter{equation}{0}
\begin{align}\label{pdf_users}
\mathbb{P}\{\mathcal{N} = \mathtt{n}\} \approx \frac{\Gamma(\mathtt{n}+c)}{\Gamma(\mathtt{n}+1)\Gamma(c)} \frac{((1-\delta) \; \mu)^\mathtt{n} (\delta \; \mu  c)^c}{((1-\delta) \; \mu+\delta \; \mu c)^{\mathtt{n}+c}},
\end{align}
where $c=3.575$ is a constant related to approximate the PDF of the PPP Voronoi cell area in $\mathbb{R}^2$. Let $ \tilde{\delta}=(1-\delta)/\delta$, then \eqref{pdf_users} can be rewritten as:
\begin{align}\label{pdf_users1}
\mathbb{P}\{\mathcal{N} = \mathtt{n}\} \approx \frac{\Gamma(\mathtt{n}+c)}{\Gamma(\mathtt{n}+1)\Gamma(c)} \left( \frac{\tilde{\delta} }{ \tilde{\delta} +c}\right)^\mathtt{n} \left( \frac{c }{ \tilde{\delta} +c}\right)^c.
\end{align}
It is worth noting that the D2D cluster size depends only on the  {\rm CH} selection probability ($\delta$) and is independent of the intensity of the devices $\mu$. As such, $\delta$ is a key performance factor that the protocol designers can use to optimize the delay. From \eqref{pdf_users1}, it can be shown that average cluster size $\mathbb{E}\left[\mathcal{N}\right]= \tilde{\delta}$.

\newcounter{TepEquCoun}
\setcounter{TepEquCoun}{\value{equation}}
\setcounter{equation}{8}
\begin{figure*}[!b]
	\hrulefill
	\small
	\begin{align}\label{eq:CDF_I}
	\!F_{\mathcal{I}}\left(x \right)= &\frac{1}{2}\!-\!\frac{1}{\pi}\!\! \int\limits_{0}^{\infty} \! \frac{1}{t} \text{Im} \left\{ \exp\left\{ -j \;t \;x  \right\}{\exp\left\{\!\!\ \! - \! 2 \; \delta \; \alpha \; \rho^{2/\eta}_{\text{L}} \! \! \! \int\limits_{\rho_{\text{L}}^{\frac{-1}{\eta}}}^{\infty}\! \! \!\left(\!1\! - \! \frac{\exp\{j \; \tau \;t \;z^{-\eta }\}}{-j\;t\;z^{-\eta} +1}\!\right)z \; dz \right\}}{ \left(  1+  \frac{\delta \; \alpha ((1-j\;t\;\rho_{\text{L}})-\exp\{j\tau \;t\;\rho_{\text{L}}\})} {(1-j\;t\;\rho_{\text{L}}) c}\right)^{-c} } \right\} dt. 
	\end{align}
	\normalsize
\end{figure*}

\vspace{-2mm}
\subsection{{\rm RA} Success Probability } \label{Sec:RA_success}
\vspace{-1mm}
Let  $P_{\text{{\rm RA}}}=\mathbb{P}\left\{ \Upsilon_{\text{{\rm RA}}}\! > \theta_{\text{{\rm RA}}} \right\}$  denote the probability that the {\rm CH}'s {\rm RA} attempt over the LTE interference is successful. As such, the SINR for the {\rm RA} ($\Upsilon_{\text{{\rm RA}}}$) can be computed as follows: 
\setcounter{equation}{2}
\begin{align}
\Upsilon_{\text{{\rm RA}}}&=\frac{\rho_{\text{L}} \; \mathtt{h}_{\circ}}{{\sigma^2+\underset{ m  \in \tilde{{\Phi}}}{\sum} \mathtt{P}_{\text{{\rm RA}}m} \ \mathtt{h}_m \ \mathtt{R}_m\!^{-\eta}}},
\end{align}
where  $\mathtt{h}_{\circ}$ represents the LTE channel gain between the test {\rm CH} and its associated {\rm {\rm BS}}. $\sigma^2$ denotes the noise power and $\eta$ denotes the path-loss exponent. The set $\tilde{\Phi}$ contains all the interfering {\rm CH}s that are simultaneously performing {\rm RA} over the same ZC code, which may contain intra-cell and inter-cell interferes due to the uncoordinated nature of the {\rm RA}.  $\mathtt{P}_{\text{{\rm RA}}m}, \ \mathtt{h}_m,\  \text{and} \ \mathtt{R}_m $ represent, respectively, the transmit power, the channel gain, and the distance between the interfering {\rm CHs} and the associated {\rm {\rm BS}} of the test {\rm CH}.  Due to the uncoordinated nature of the {\rm RA}, there are two possible sources of interferences, namely the intra-cell interference and the inter-cell interference. Using stochastic geometry,  we  characterize the intra-cell and inter-cell interference on a test device via the LT of their probability density functions (PDFs). Then, the obtained LTs are used to derive  the {\rm RA} success probability, which is characterized in terms of {\rm CH} selection probability ($\delta$) and the device-to-{\rm BS} ratio ($\alpha$), i.e., $\alpha=\mu/\lambda$ .

{From the independent thinning property of the PPP~\cite{martin_book}, the {\rm CH}s interfering on the same ZC code constitute a PPP  with intensity $\frac{ \delta \; \mu}{{n_{Z} }}$, where  $n_{ Z}$ is the number of available ZC codes. Consequently,  the average number of {\rm CH}s that may use the same ZC code per {\rm BS} is given by ${\tilde{\alpha}} =\frac{\delta \; \mu}{\lambda \; n_Z} $.} By the PPP assumption, the number and locations of the points in disjoint areas are independent.  Consequently, the intra-cell and inter-cell interference are independent. Exploiting this fact, the success probability for each of the clustering schemes is given in the sequel.
\subsubsection{RBC Scheme} The RA success probability for the RBC can be expressed as 
\begin{align}
P_{\text{{\rm RA}}} =&  \mathbb{P} \left\{ \frac{\rho_{\text{L}}\; \mathtt{h}_{\circ}}{\sigma^2 + \mathcal{I}_{\text{In}}+\mathcal{I}_{\text{Out}}}>\theta_{\text{{\rm RA}}} \right\}  \notag \\
\overset{(a)}{=} & \exp\left\{- \frac{\sigma^2 \theta_{\text{{\rm RA}}}}{\rho_{\text{L}}} \right\} \mathscr{L}_{\mathcal{I}_{\text{In}}} \left(\frac{\theta_{\text{{\rm RA}}}}{\rho_{\text{L}}} \right) \mathscr{L}_{\mathcal{I}_{\text{Out}}} \left(\frac{\theta_{\text{{\rm RA}}}}{\rho_{\text{L}}} \right){,}
\label{SINR_RA_1}
\end{align}
where ${\mathcal{I}_{\text{In}}}$ is the intra-cell interference and ${\mathcal{I}_{\text{Out}}}$ is the inter-cell interference. Note that $(a)$ in \eqref{SINR_RA_1} follows from the exponential distribution of $\mathtt{h}_{\circ}$ \cite{survey_h}. The {\rm {\rm RA}} access success probability in \eqref{SINR_RA_1} is characterized with the following theorem.

\begin{theorem}\label{lemma_RA_success}
	The {\rm {\rm RA}} access success probability  in a PPP network and Random-Based Clustering where each {\rm CH} employs full path-loss inversion power control is given by:
	\small
	\begin{align}\label{eq:RA_success}
	{P_{\text{{\rm RA}}}} &\; {\approx} \; \frac {\exp\left\{- \frac{\sigma^2 \theta_{\text{{\rm RA}}}}{\rho_{\text{L}}} -  \frac{2 \; \delta\;  \alpha \; \theta_{\text{{\rm RA}}}}{n_Z (\eta-2) } \;{}_2F_1\left(1,1-\frac{2}{\eta},2-\frac{2}{\eta},-\theta_{\text{{\rm RA}}}\right) \right\}}{ \left(  1+  \frac{\delta \;\alpha \;\theta_{\text{{\rm RA}}} }{n_Z (1+\theta_{\text{{\rm RA}}}) c}\right)^{c} } ,
	\end{align}
	\normalsize
	where $c=3.575$ is a constant related to the approximate PDF of the PPP Voronoi cell area in $\mathbb{R}^2$.
\end{theorem}

\begin{proof}
	Similar to \cite[Lemma 1]{Gharbieh_tcom}, where \ref{eq:RA_success} is not exact  due to ignoring the spatial correlations among the transmission powers of the {\rm CH}s as mentioned earlier in Approximation~\ref{app1}.
\end{proof}
For the special case of $\eta=4$, which is a typical path loss exponent for urban outdoor environment, \eqref{eq:RA_success} reduces to:
\begin{align} \label{eq:Out10}
{P_{\text{{\rm RA}}}} \;  {\approx} \; \frac{\exp\left\{- \frac{\sigma^2 \theta_{\text{{\rm RA}}}}{\rho_{\text{L}}} -\frac{\delta \alpha}{n_Z} \sqrt{\theta_{\text{{\rm RA}}}}\arctan\left({\sqrt{\theta_{\text{{\rm RA}}}}}\right)   \right\}}{ \left(  1+  \frac{\delta \alpha \theta_{\text{{\rm RA}}}  }{n_Z (1+\theta_{\text{{\rm RA}}})c}\right)^{c} }.
\end{align}
The expression {in \eqref{eq:Out10}} gives the {\rm RA} success probability in terms of the elementary arctan function instead of the computationally complex  Gaussian hypergeometric function.

\vspace{2mm}
\subsubsection{CGBC Scheme} The {\rm {\rm RA}} success probability for the CGBC scheme can be expressed as: 
\small
\begin{align}\label{SINR_RA_2} 
P_{\text{{\rm RA}}} =&  \mathbb{P} \left\{ \frac{\rho_{\text{L}}\; \mathtt{h}_{\circ}}{\sigma^2 + \mathcal{I}_{\text{In} \mid  \mathtt{h} > \tau}+\mathcal{I}_{\text{Out} \mid  \mathtt{h} > \tau}}>\theta_{\text{{\rm RA}}} \mid \mathtt{h} > \tau \right\}  \notag \\
\overset{(b)}{=} & \begin{cases} 
\exp\left\{- \frac{\sigma^2 \theta_{\text{{\rm RA}}}}{\rho_{\text{L}}} +\tau \right\} \mathscr{L}_{\mathcal{I}_{\text{In}} \mid  \mathtt{h} > \tau} \left(\frac{\theta_{\text{{\rm RA}}}}{\rho_{\text{L}}} \right)\\ \quad  \quad \quad \quad  \quad \quad \times \mathscr{L}_{\mathcal{I}_{\text{Out}} \mid  \mathtt{h} > \tau} \left(\frac{\theta_{\text{{\rm RA}}}}{\rho_{\text{L}}} \right),  & \frac{\theta_{\text{{\rm RA}}}(\sigma^2+\mathcal{I})}{\rho_{\text{L}}}>\tau \\
1, & \text{otherwise{,}}
\end{cases}
\end{align}
\normalsize
where ${\mathcal{I}_{\text{In} \mid  \mathtt{h} > \tau}}$ is the intra-cell interference given that the interfering {\rm CH}s have channel gain greater than $\tau$,  ${\mathcal{I}_{\text{Out} \mid  \mathtt{h} > \tau}}$ is the inter-cell interference given that the interfering {\rm CH}s have channel gain greater than $\tau$, and $\mathcal{I}$ is the total interference. Note that $(b)$ in \eqref{SINR_RA_2} follows from the exponential distribution of $\mathtt{h}_{\circ}$ \cite{survey_h}. The {\rm {\rm RA}} access success probability in \eqref{SINR_RA_2} is characterized by the following theorem.

\begin{theorem}\label{lemma_RA_success2}
	The {\rm {\rm RA}} access success probability  in a PPP network and Channel-Gain-Based Clustering where each {\rm CH} employs full path-loss inversion power control is given by
	\small
	\begin{align}\label{eq:RA_success2}
	{P_{\text{{\rm RA}}}} \;{\approx}\;  & \frac {\exp\left\{\!\!\!- \frac{\sigma^2 \theta_{\text{{\rm RA}}}}{\rho_{\text{L}}}\! +\! \tau \! - \! 2 \; \delta \; \alpha \; \theta^{2/\eta}_{\text{{\rm RA}}} \! \! \! \int\limits_{\theta_{\text{{\rm RA}}}^{\frac{-1}{\eta}}}^{\infty}\! \! \!\left(\!1\! - \! \frac{\exp\{\!-\tau y^{-\eta }\!\}}{y^{-\eta} +1}\!\right)y \; dy \right\}}{ \left(  1+  \frac{\delta \; \alpha ((1+\theta_{\text{{\rm RA}}})-\exp\{-\tau \theta_{\text{{\rm RA}}}\})} {(1+\theta_{\text{{\rm RA}}}) c}\right)^{c} } \notag \\
	\times & \bar{F}_{\mathcal{I}}\left(\frac{\tau \; \rho_{\text{L}}}{\theta_{\text{{\rm RA}}}}-\sigma^2\right)+F_{\mathcal{I}}\left(\frac{\tau \; \rho_{\text{L}}}{\theta_{\text{{\rm RA}}}}-\sigma^2\right),
	\end{align}
	\normalsize
	where $c=3.575$ is a constant related to the approximate PDF of the PPP Voronoi cell area in $\mathbb{R}^2$ and $F_{\mathcal{I}}\left(x \right)$ is the CDF of the aggregated interference which has the form of \eqref{eq:CDF_I} in the bottom of this page.
\end{theorem}

\begin{proof}
	See Appendix \ref{proof1},  where \ref{lemma_RA_success2} is not exact  due to ignoring the spatial correlations among the transmission powers of the {\rm CH}s as mentioned earlier in Approximation~\ref{app1}.
\end{proof}

\vspace{-2mm}
\subsection{D2D Transmission Success Probability} \label{Sec:D2D}
\vspace{-1mm}
On the D2D links, each {\rm CH} aggregates the uplink scheduling requests originated from its associated {\rm CM}s TDMA scheduling. Hence, inter-cluster interference and noise are the only two channel impairment for D2D transmission. Therefore, the transmission SINR for the D2D relaying can be computed as follows:
\setcounter{equation}{9}
\begin{align}
\Upsilon_{\text{C}}&=\frac{\rho_{\text{C}} \;  \mathtt{h}_\ast }{{\sigma^2+\underset{i  \in {{\Phi}}}{\sum} \mathtt{P}_{\text{D}i} \ \mathtt{h}_i \ \mathtt{R}_i\!^{-\eta}}},
\end{align}
where  $\mathtt{h}_\ast$ represents the D2D channel gain between the test {\rm CM} and its associated {\rm CH}. The set $\Phi$ contains all the interfering {\rm CMs} that transmit simultaneously over the same frequency, which are one {\rm CM} per cluster due to the TDMA scheduling.  $\mathtt{P}_{\text{D}i}, \ \mathtt{h}_i,$ and $\mathtt{R}_i $ represent the transmit power, the channel gain, and the distance between the interfering {\rm CMs} and the associated {\rm CH} of the test {\rm CM}. To characterize $\Upsilon_{\text{C}}$, we follow a similar methodology described for $\Upsilon_{\text{RA}}$, while accounting for the fact that each {\rm CH} has a single active {\rm CM} to serve at a given time instant. { Hence, the  intensity of interfering {\rm CM}s on each channel is equal to the intensity of the {\rm CH}s that selected the same frequency, i.e., intensity of interfering {\rm CM}s is $\frac{\delta \; \mu}{k}$, where $k$ is the number of frequencies available for D2D transmission.} Consequently, the D2D success probability $P_{\text{C}}=\mathbb{P}\left\{ \Upsilon_{\text{C}}\! > \theta_{\text{C}}  \right\}$ is characterized by the following theorem.

\begin{theorem}\label{lemma_Tx_success}
	The probability of successful uplink scheduling request over a PPP D2D link where each {\rm CM} employs full path-loss inversion power control, can be expressed as
	\small
	\begin{align}\label{eq:Out2}
	{P_{\text{C}}} &\; {\approx} \; \exp\left\{- \frac{\sigma^2 \theta_{\text{C}}}{\rho_{\text{C}}} -  \frac{2  \theta_{\text{C}}}{k (\eta-2) } \;{}_2F_1\left(1,1-\frac{2}{\eta},2-\frac{2}{\eta},-\theta_{\text{C}}\right) \right\}. 
	\end{align}
	\normalsize
\end{theorem}

\begin{proof}
	{Similar to \cite[Theorem 1]{elsawy2014stochastic}, where \eqref{eq:Out2} is not exact  due to ignoring the spatial correlations among the transmission powers of the {\rm CM}s as mentioned earlier in Approximation~\ref{app1} and  due to approximating the inter-cluster interfering {\rm CM}s with a PPP as mentioned in Approximation~\ref{app2}.}
\end{proof}

For the special case of $\eta=4$, \eqref{eq:Out2} reduces to:
\begin{align} \label{eq:Out100}
{P_{C}} \; {\approx} \; \exp\left\{- \frac{\sigma^2 \theta_{\text{C}}}{\rho_{\text{C}}} -\frac{\sqrt{\theta_{\text{C}}}}{k} \arctan\left({\sqrt{\theta_{\text{C}}}}\right)   \right\}.
\end{align}
It is worthwhile to mention that the approximations in  \eqref{eq:RA_success}, \eqref{eq:RA_success2}, and \eqref{eq:Out2} are mandatory for tractability and are common in stochastic geometry analysis for uplink systems~\cite{uplink2_jeff, uplink_alamouri, marco_uplink}.

\subsection{ Channel Access Delay ($D$) } \label{Sec:delay}
The channel access delay ($D$)  is defined as the average number of time slots required by a device before the uplink transmission is successfully scheduled. In the depicted communication system, where the {\rm CMs} sends their requests to the {\rm CH} via a TDMA schedule, $D$ is given by
\begin{align} \label{ave_retransmission}
D & {\approx} \frac{1}{P_{\text{{\rm RA}}}} + \frac{\mathbb{E}\left[\mathcal{N}\right]}{P_{\text{C}}},
\end{align}
\noindent  which is derived by modeling the trials for both of the {\rm RA} and the D2D transmissions by geometric random variables. In addition, the mean number of {\rm CMs} associated to a {\rm CH} ($\mathbb{E}\left[\mathcal{N}\right]$) is introduced here to take into account the TDMA scheduling. It is worth to highlight that \eqref{ave_retransmission} is not exact due to ignoring the negligible temporal correlation of the transmission success probabilities as mentioned in Approximation~\ref{app3}.

\vspace{-2mm}
\section{Numerical Results} \label{Results}
\vspace{-1mm}
\begin{figure}[t!]	
	\begin{center}
		\subfigure[$\mu=160 $, $\delta= .1$ ]{\label{fig:out1}\includegraphics[width=2.78in]{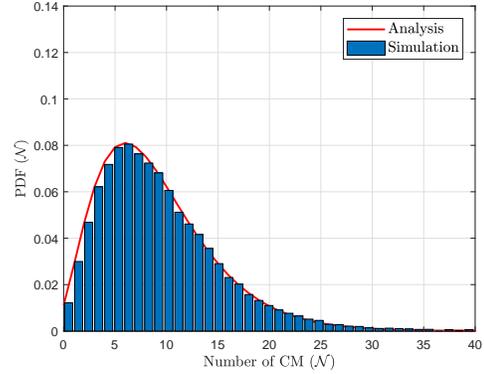}}\vspace{-2mm}
		\subfigure[$\mu=640 $, $\delta= .1$]{\label{fig:out2}\includegraphics[width=2.78in]{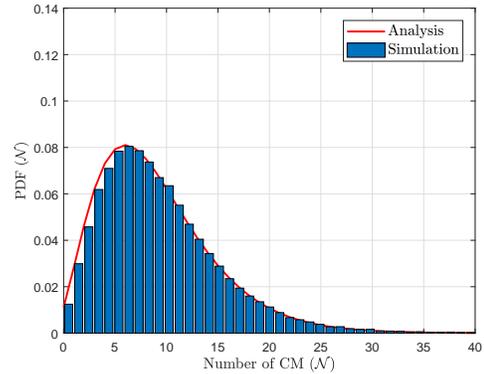}}\vspace{-2mm}
		\subfigure[$\mu=160 $, $\delta= .15$ ]{\label{fig:out2}\includegraphics[width=2.78in]{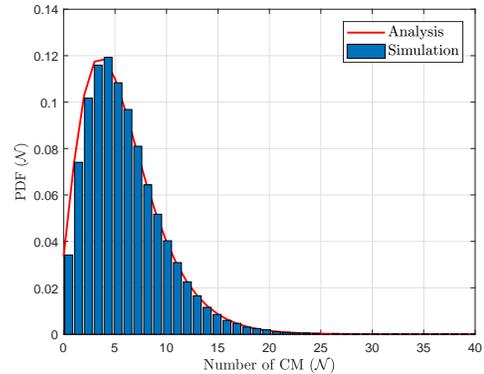}}\vspace{-2mm}
		\subfigure[$\mu=640 $, $\delta= .15$]{\label{fig:out2}\includegraphics[width=2.78in]{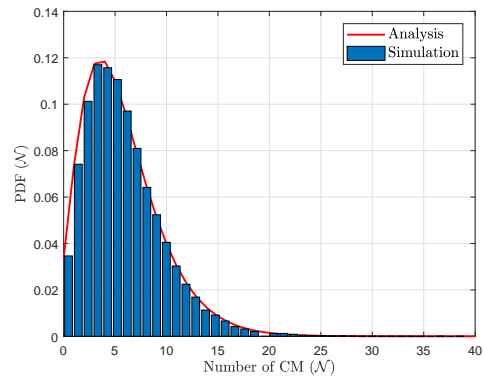}}
	\end{center}
	\vspace{-5mm}
	\caption{:  The PDF of the number of {\rm CM}s associated to a {\rm CH} for for $\mu=160 \text{ and } 640 $ UE/km$^2$ for different values of CH probability $\delta= .1$ and $.15$. }
	\vspace{-5mm}
	\label{fig:PDF_USER}
\end{figure}

This section first validates the developed model via independent Monte Carlo simulations. Then, selected numerical results are presented to assess and compare the performance of the RBC and CGBC schemes. In each simulation run, the BSs and the devices are realized over a 100 km$^2$ area via independent PPPs and the collected statistics are taken for devices located within 1 km from the origin to avoid the edge effects. First, we examine the D2D cluster size in \eqref{pdf_users1}. Fig.~\ref{fig:PDF_USER} shows the PDF of the associated {\rm CM}s for each {\rm CH} for $\mu=160 \text{ and } 640 $ UE/km$^2$ for different values of CH selection probability $\delta=0 .1$ and $0.15$. Fig.~\ref{fig:PDF_USER} supports the remark given in Section \ref{sec:PDF} about the independence of the PDF of $\mathcal{N}$ from the devices intensity. This can be explained as follows: The cluster's geographical footprint is represented by a Voronoi cell whose average area is $\bar{d} =\frac{1}{\delta \mu}$. The average geographical footprint $\bar{d}$ of the cluster shrinks as more devices are elected as {\rm CH}s. However, the intensity of CMs also increases such that the cluster size in terms of number of CMs stays constant. Fig.~\ref{fig:PDF_USER} also shows that as $\delta$ increases, the D2D cluster size becomes smaller, and hence, the PDF of $\mathcal{N}$ is pushed to have a smaller mean.

\begin{figure}[t!]
	
	\begin{center}
		\subfigure[$\alpha=16$ device/{\rm BS}.]{\label{fig:out1}\includegraphics[width=3.2in]{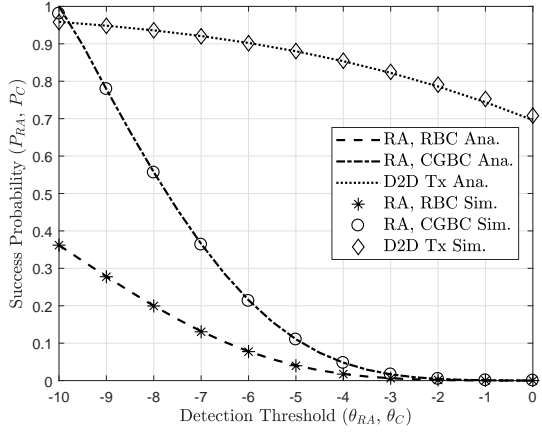}}
		\subfigure[$\alpha=64$ device/{\rm BS}.]{\label{fig:out2}\includegraphics[width=3.2in]{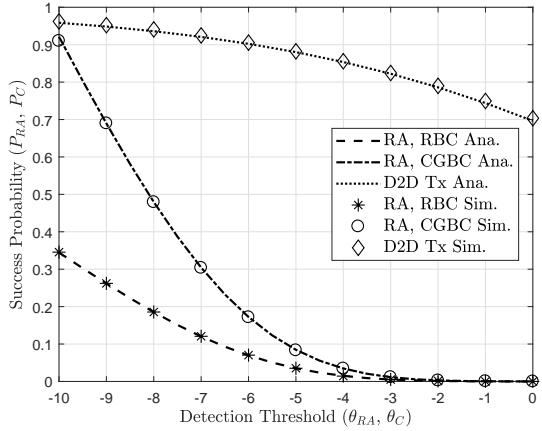}}
	\end{center}
	\vspace{-5mm}
	\caption{: The {\rm RA} and D2D transmission success probability as a function of $\theta_{\text{{\rm RA}}}$ and $\theta_{\text{C}}$ for {\rm CH} selection probability $\delta=0.35$.}
	\label{fig:success_thet}
\end{figure}

{Fig.~\ref{fig:success_thet} depicts the {\rm RA} transmission success probabilities for RBC and CGBC schemes along with the D2D transmission success probability. The simulation parameters are as follows; $\mu= 160$ and $640$ device/km$^2$, $\lambda=10$ BS/km$^2$, and equivalently, $\alpha= 16$ and $64$ device/BS. {\rm CH} selection probability $\delta=0.35$, $n_z$= 64 code per BS, $\rho_{\text{C}}=-80$ dBm,  $\rho_{\text{RA}}=-100$ dBm, noise power $\sigma^2=-90$ dBm,    number of frequencies available for D2D transmission  $k=3$, and path-loss exponent $\eta=4$.  It is important to note the close match between the analysis and simulation results which validates the developed mathematical framework and Approximations \ref{app1}-\ref{app2}.}

\begin{table}[t!]
	\vspace{5mm}
	\centering
	\footnotesize
	\renewcommand{\arraystretch}{1.3}
	\begin{tabular}{||p{.87cm}|p{3.5cm}|p{3.067cm}|}  
		\hline  \textbf{Notation} & \textbf{Description} & \textbf{Value}\\
		\hline $\mu$ & device densities &$160$ and $640$ devices/km$^2$ \\
		\hline $\lambda$  & {\rm BS} density & $10$  {\rm BS}/km$^2$ \\
		\hline ${n_{Z}}$ & number of ZC codes dectitated for random access & $64$ code per {\rm BS}\\
		\hline $\delta$ & {\rm CH} probability  & $[.1,1]$ \\
		\hline $\theta_{\text{C}} $ &  Detection threshold for \newline successful D2D transmission  & $-7$ dB\\
		\hline $\theta_{\text{{\rm RA}}} $ & Detection threshold for \newline successful {\rm {\rm RA}}  & $-7$ dB\\
		\hline $\rho_{\text{c}}$  & Power control parameter for D2D  & $ -80$ dBm  \\
		\hline $\rho_{\text{L}}$  & Power control parameter for {\rm {\rm RA}} &  $-100$ dBm \\
		\hline $\sigma^2$ &  noise power & $-90$ dBm\\
		\hline $k$ &  number of frequencies \newline available for D2D transmission & $3$ \\
		\hline $\eta$ & path-loss exponent & $4$ \\		
		\hline
	\end{tabular}
	\caption{: Numerical Evaluation Parameters}\label{table_parameters}
\end{table}

At this point of the discussion, we look into the channel access delay as a key performance metric. The conventional {\rm RA} procedure is used as a benchmark to evaluate the performance of the proposed D2D clustering procedure. It is worth mentioning that the conventional {\rm RA} procedure is a special case of the RBC D2D clustering by setting $\delta=1$, also it is a special case of the CGBC D2D clustering by setting $\tau=0$. The parameters used in this section are summarized in Table~\ref{table_parameters}.

\begin{figure}[t!]
	
	\begin{center}
		\subfigure[$\alpha=16$ device/{\rm BS}.]{\label{fig:out1}\includegraphics[width=3.2in]{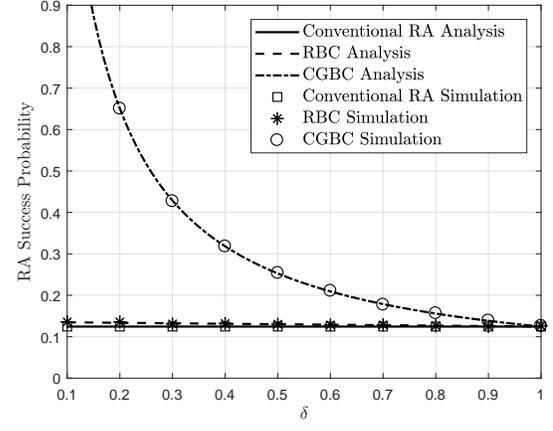}}
		\subfigure[$\alpha=64$ device/{\rm BS}.]{\label{fig:out2}\includegraphics[width=3.2in]{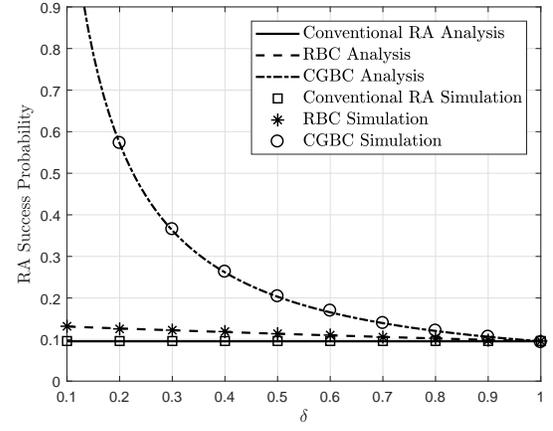}}
	\end{center}
	\vspace{-5mm}
	\caption{: The {\rm RA} success probability as a function of {\rm CH} selection probability $\delta$.}
	\label{fig:success}
\end{figure}

 Fig.~\ref{fig:success} and Fig.~\ref{fig:Delay} show respectively the {\rm RA} success probability and the delay as a function of $\delta$. As expected, Fig.~\ref{fig:success} shows that the {\rm RA} success probability for both the D2D clustering and the conventional cases decreases as $\alpha$ grows. And in turn, the delay increases as $\alpha$ increases. Also, Fig.~\ref{fig:success} shows that as the {\rm CH} selection probability $\delta$ increases the {\rm RA} success probability for D2D clustering decreases. This is due to the fact that as $\delta$ increases, more {\rm CH}s are eligible to perform an {\rm RA} procedure over the LTE interface. Therefore, both of the inter-cell and intra-cell interference increases leading to lower {\rm {\rm RA}} success probability ($P_\text{{\rm RA}}$). Furthermore, the results show that CGBC D2D clustering offers the highest {\rm {\rm RA}} success probability.

\begin{figure}[t!]
	
	\begin{center}
		\subfigure[$\alpha=16$ device/{\rm BS}.]{\label{fig:out1}\includegraphics[width=3.2in]{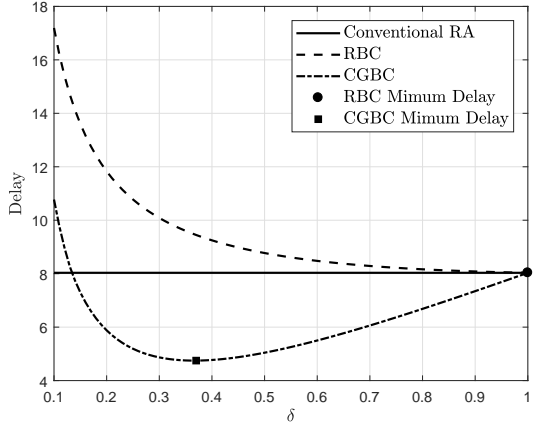}}
		\subfigure[$\alpha=64$ device/{\rm BS}.]{\label{fig:out2}\includegraphics[width=3.2in]{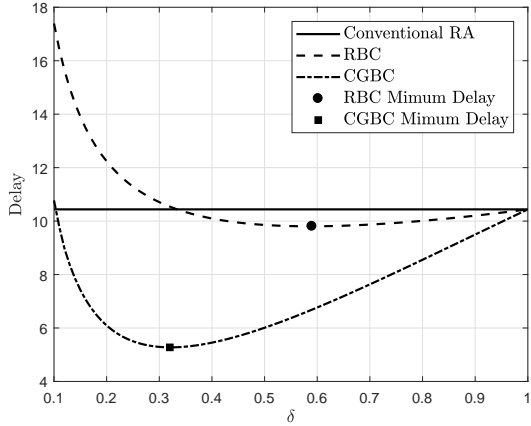}}
	\end{center}
	\vspace{-5mm}
	\caption{:  The average delay needed for successful {\rm RA} as a function of {\rm CH} selection  probability $\delta$. }
	\label{fig:Delay}
\end{figure}

For insightful conclusions,  Fig.~\ref{fig:success} and Fig.~\ref{fig:Delay} should be considered jointly. While one case may be favorable from the {\rm RA} success  probability perspective, it may be invoking too much delay and an adverse impact on the average waiting time for a successful {\rm RA}. Fig. \ref{fig:Delay} gives an interesting insight by comparing performance at two device densities. For RBC D2D clustering, the {\rm RA} performance will not gain any benefit at low device densities. Regardless how aggressive $\delta$ is, the conventional case always offers lower channel access delay. From a mere {\rm RA} perspective, it is simply just not worth it to use D2D clustering. However, the RBC D2D clustering starts to pay off at high intensities as shown in Fig.\ref{fig:Delay}. On the other hand, the CGBC offers an improvement over the conventional {\rm RA} even for low device intensities as shown Fig. \ref{fig:Delay}.

It is straightforward to notice the trade-off between the delay and the {\rm CH} selection probability $\delta$ or equivalently $\tau$. Specifically, $\delta$ has a two-fold effect on the delay as can be inferred from Eq. \eqref{ave_retransmission}. First, the larger the $\delta$ the more device are eligible to be {\rm CH}s, and hence, the smaller cluster size in terms of the number of the associated {\rm CM}s. Therefore, smaller number of {\rm CM}s results in a shorter delay in the  TDMA scheduled transmission. As such, the right-hand term of Eq. \eqref{ave_retransmission} decreases because the numerator ($\mathcal{N}$) decreases while the denominator  ($P_C$) is not intact with the change of $\delta$. Second, it is worthy recalling however that larger $\delta$ means a larger number of {\rm CH}s, which leads to degraded success probability of the {\rm {\rm RA}} over the LTE PRACH interface.  In more precise terms, the first term of Eq. \eqref{ave_retransmission} increases as $\delta$ increases. To sum up, the first term of Eq. \eqref{ave_retransmission} is a negative monotone in $\delta$ while the second is positive monotone. Consequently, $\delta$ can be optimized to achieve a minimum delay. For example,  employing RBC D2D clustering and at  $\alpha=64$ device/{\rm BS} with a value of ${\delta^*}=0.59$  minimizes the delay with a rate of $6 \%$ when compared to the conventional {\rm RA} case. On the other hand, for CGBC D2D clustering and at  $\alpha=16$ device/{\rm BS}  with a value of ${\tau^*}=-0.025$ dB and a corresponding value of  ${\delta^*}=0.37$ would minimize the delay at a reduction rate of $40 \%$ when compared to the conventional {\rm RA} case. While at $\alpha=64$ device/{\rm BS}, ${\tau^*}=0.567$ dB with a corresponding value of  ${\delta^*}=0.32$ would minimize the delay at a reduction rate of $49 \%$. 

\begin{figure}[t!]
	\begin{center}
		\includegraphics[width=3.5 in]{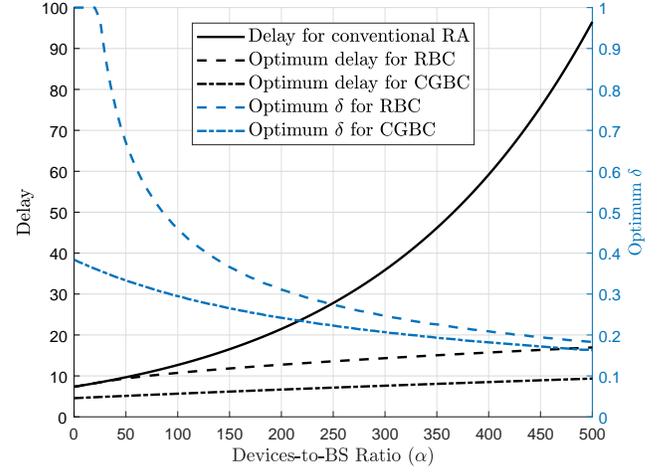}	
	\end{center}
	\vspace{-3mm}
	\caption{:  Optimum delay and the corresponding values of ${\delta^*}$  as a function of the devices-to-{\rm BS} ratio ($\alpha$).  }
	\vspace{-2mm}
	\label{fig:optimum}
\end{figure}

Fig.~\ref{fig:optimum} depicts the optimum delay and the corresponding value of the CM selection probability ($\delta$) for both RBC and CBGC clustering schemes. The result shows that the RBC starts to offer a reduction in the {\rm RA} access delay when devices-to-{\rm BS} ratio becomes larger than 50. However, the CBGC always offers an enhancement over the conventional {\rm RA}. Another insightful observation from Fig.~\ref{fig:optimum} is that as $\alpha$ increases, $\delta$ decreases. This behavior is mainly due to the fact that the {\rm RA}  congestion  over the LTE interface is more  critical when the devices intensity increases compared with the delay that stems from the TDMA transmission within the D2D cluster.

\begin{figure}[t!]
	\begin{center}
		\includegraphics[width=3.5 in]{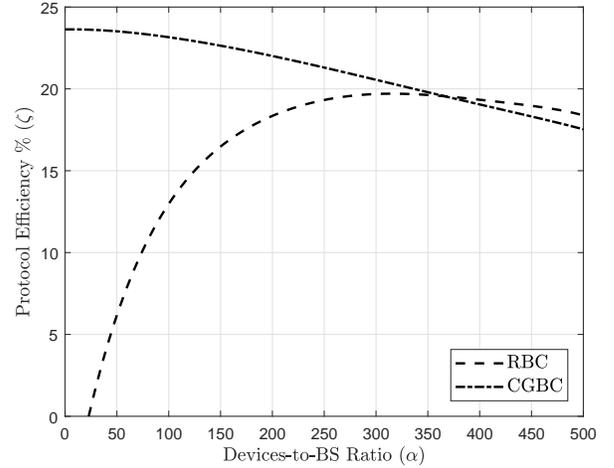}	
	\end{center}
	\vspace{-3mm}
	\caption{: Protocol efficiency for optimum delay reduction rate as a function of the devices-to-{\rm BS} ratio ($\alpha$).  }
	\label{fig:eff_optimum}
\end{figure}

It should be mentioned, however, that the improved performance of the CGBC scheme comes at the cost of higher protocol overhead. That can be justified by the lower optimal CH selection probability, which leads to higher cluster population for the CGBC scheme, and in turn, larger signaling in the cluster formation process. As such, there is a need to quantify the protocol efficiency which can be defined by the delay reduction rate over the protocol overhead. We consider the average cluster size $\mathbb{E}\left[\mathcal{N}\right]$  as a measure for the protocol overhead. Therefore the protocol efficiency ($\zeta$) can be calculated as follows: 
\begin{align} \label{eq:protocol_eff}
\zeta = \frac{D(\delta={\delta^*})-D(\delta=1)}{D(\delta=1)\mathbb{E} \;\left[\mathcal{N}\right]}  \times  100 \%,
\end{align}
where $D(\delta)$ is the delay in \eqref{ave_retransmission} as a function of the {\rm CH} selection probability $(\delta)$.

Fig.~\ref{fig:eff_optimum} shows the protocol efficiency for the optimal delay in Fig.~\ref{fig:optimum}. The figure shows that even though the CGBC provides a lower delay, the protocol efficiency is lower when the devices intensity scales. Specifically, it is more rewarding in terms of protocol efficiency to follow the RBC scheme when the devices intensity goes beyond $370$ devices/{\rm BS}.

	\vspace{-2mm}
\section{Implementation, Issues, $\&$ remedies} \label{sec:Imp}
	\vspace{-1mm}
The foreseen gain that Fig.~\ref{fig:optimum} depicts can be best achieved through an automated self-optimization algorithm. Such an algorithm can be implemented through a back-end script at the network core whose goal is to estimate network parameters and calculate the optimum value of {\rm CH}s selection probability ($\delta$) for RBC and CGBC schemes. The analytical results in Section~\ref{sic:Performance_Analysis} show that the performance of the proposed D2D clustering schemes depends on three estimated parameters, namely, the intensity of {\rm BS}s, the devices intensity, and the path-loss exponent. The intensity of {\rm BS}s may be the easiest parameter to estimate as the number of BSs in a geographic area is available to the operators, then the intensity of BSs can be estimated accordingly. The devices intensity, on the other hand, can be estimated through event-triggered reporting for the association table. As such, the {\rm CH} reports to the {\rm BS} any change occurred to its associated {\rm CM}s, then the {\rm BS}, in turn, reports this change to the core network. The core network can estimate the current device intensity and then broadcast the optimum value of $\delta$ in a reclustering order. {Moreover, for the CGBC scheme, the CHs are required to report to the BS if the estimated channel gain goes below the threshold $\tau$. In this case, the BS can broadcast  a reclustering order to maintain the performance edge.} Lastly, the path-loss exponent estimation can be done by a self-estimator that only requires collecting multiple Received Signal Strength (RSS) as in \cite{6981955}  which can be executed easily due to its independence. The pseudo code for the back-end self-optimization script is shown in {\bf{Algorithm \ref{back_algo}}}, {where the optimization problems in \eqref{optimization} and \eqref{optimization1} can be solved via a one-dimensional line search with an initial uncertainty range of $I_{\circ}\in (0,1]$. One of the algorithms that can be used to solve \eqref{optimization} and \eqref{optimization1}  is the golden-section search. The number of golden-section search iterations ($n$) that achieves an accuracy of ($\epsilon$) for the  can be estimated by \cite{Antoniou_optimization}
\begin{align}
n \geq \log_{K}\frac{I_{\circ}}{\epsilon},
\end{align} 
\vspace{-1mm}
where $K=(1+\sqrt{5})/2$ is the golden ratio constant. For example, the golden-section search achieves an accuracy of $\epsilon=10^{-6}$ with $n=29$ iterations and only $30$ function evaluations. The brute-force method, on the other hand, requires $10^{\;6}$ function evaluations to find a $\delta^*$ with the same accuracy.}

\begin{algorithm}[t]
	\caption{\!\!: Pseudo code for the back-end self-optimization script.}\label{euclid}
	\begin{algorithmic}[1]
		\State Estimate $\lambda, \mu, \eta$.
		\State Solve the optimization problem of $D$ in \eqref{ave_retransmission} for RBC scheme as:
		\vspace{-4mm}
		\begin{equation}
		\begin{aligned}
		& \underset{\delta}{\text{minimize}}
		& &  D_{\text{RBC}}=\frac{1}{P_{\text{{\rm RA}}}} + \frac{\tilde{\delta}}{P_{\text{C}}}. \\
		& \text{subject to}
		& & 0 < \delta \leq 1 
		\end{aligned}
		\label{optimization}
		\end{equation} where $P_{\text{{\rm RA}}}$ can be evaluated by \eqref{eq:RA_success}. 
		\State Solve the optimization problem of $D$ in \eqref{ave_retransmission} for CGBC scheme as:
			\vspace{-4mm}
		\begin{equation}
		\begin{aligned}
		& \underset{\delta}{\text{minimize}}
		& &  D_{\text{CGBC}}=\frac{1}{P_{\text{{\rm RA}}}} + \frac{\tilde{\delta}}{P_{\text{C}}}. \\
		& \text{subject to}
		& & 0 < \delta \leq 1 
		\end{aligned}
		\label{optimization1}
		\end{equation} where $P_{\text{{\rm RA}}}$ can be evaluated by \eqref{eq:RA_success2}, and $\tau=-\ln(\delta)$.
		\State Compare $D_{\text{RBC}}$ and $ D_{\text{CGBC}}$  and select the scheme that results in lower $D$.
		\State Return the optimum $\delta$ and the selected scheme.
	\end{algorithmic}
	\label{back_algo}
\end{algorithm}

However, incentivizing and commercializing D2D clustering have been dwelling in a slightly stagnant state for some time due to some practical concerns. First, one of the major implementation challenges of the proposed scheme is that the {\rm CM} may fall in a BS footprint different than the {\rm CH} is associated to. Since the uplink resources are better granted to the devices by the closest BS, the core network is in the best position to process the uplink resource scheduling such that each device is granted uplink resources from the nearest BS. Second, the use of D2D clustering network to reroute uplink scheduling requests entails fairness issues regarding battery depletion rates of the {\rm CH}s. Furthermore, the proposed D2D clustering entails low-layer modifications to the protocol stack something that needs to be taken into standardization meetings. 

\vspace{-2mm}
\section{Conclusions}\label{sec:Conclusions}
\vspace{-1mm}
This paper introduces a self-organized D2D clustering scheme to relieve the congestion on the {\rm RA} resources in massively loaded networks. Two D2D clustering schemes are studied, namely, RBC and CGBC. The results show that the RBC scheme offers no delay reduction at low device densities and hence it is preferable to follow the conventional access model. As the device intensity grows, the RBC starts to offer reduced delay and there is an optimal value of {\rm CH}s selection probability ($\delta$) that minimizes the delay. To maintain the performance edge of the RBC, the {\rm BS} has to revisit the clustering relationships whenever the intensity changes. On the other hand, the CGBC offers significant performance gains when compared to both the conventional and RBC schemes even for low device intensities. However, the gain edge comes at the cost of higher overhead due to the larger cluster size, and hence, larger signaling in the clustering process. As such the two schemes offer a trade-off between complexity and performance. To this end, a self-optimization algorithm to execute the D2D clustering is presented. We also highlight a few remedies and recommendations for practical implementation. 

\appendix
\subsection{Proof of Theorem \ref{lemma_RA_success2}. }\label{proof1}
Note that the nearest {\rm BS} association and the employed power control enforce the following two conditions; (i) the intra-cell interference from an interfering device is equal to $\rho_{\text{L}}$, and (ii) the inter-cell interference from any interfering device is strictly {less than} $\rho_{\text{L}}$. The aggregated inter-cell interference received at the {\rm BS} is obtained as:
\begin{align}\label{eq:AppA_3}
\mathcal{I}_{\text{out}\mid  \mathtt{h}>\tau}=  \sum\limits_{m \in \tilde{{\Phi}} } \mathbbm{1}_{{\{}\mathtt{P}_{\text{{\rm RA}}m} \mathtt{R}_m\!^{-\eta}<\rho_{\text{L}}\}}\mathtt{P}_{\text{{\rm RA}}m} \ \mathtt{h}_m \ \mathtt{R}_m\!^{-\eta}.
\end{align}
 Approximating the set of interfering devices by a PPP with independent transmit powers, the Laplace Transform of \eqref{eq:AppA_3} can be approximated as \eqref{eq:app1}.
\begin{align} \label{eq:app1}
\mathscr{L}_{\mathcal{I}_{\text{out}\mid \mathtt{h}>\tau}}(s) \approx \exp \Bigg\{ -2\pi \; \delta\; \alpha \lambda \; \mathbb{E}_{\mathcal{P}_{\text{{\rm RA}}}}\left[\mathcal{P}_{\text{{\rm RA}}}^{\frac{2}{\eta}} \;   \right] \notag \\  \times \int\limits_{( \rho_{\text{L}} )^{\frac{-1}{\eta}}}^{\infty}\left(1- \frac{\exp\{-\tau s \; y^{-\eta }\}}{s\;y^{-\eta} +1}\right)\;y\; dy   \Bigg\}.
\end{align}
\noindent  The LT is obtained by using the probability generating function (PGFL) of the PPP \cite{ martin_book} and following \cite{elsawy2014stochastic}, where the LT is obtained by substituting the value of $\mathbb{E}_{\mathcal{P}_{\text{{\rm RA}}}}\left[\mathcal{P}_{\text{{\rm RA}}}^{\frac{2}{\eta}}\right]$ from [Lemma 1,\cite{elsawy2014stochastic}].
The Intra-cell interference conditioned on the number of neighbors is given by:
\begin{align}\label{eq:AppA_1}
\small
\mathcal{I}_{\text{in}\mid \mathtt{h}>\tau , \mathtt{n}} = \displaystyle{\sum\limits_{i=1}^\mathtt{n} \rho_{\text{L}} \mathtt{h}_i}.
\normalsize
\end{align}
The Laplace Transform of \eqref{eq:AppA_1} is obtained as:
\begin{align}
\small
\mathscr{L}_{\mathcal{I}_{\text{in}\mid \mathtt{h}>\tau , \mathtt{n}}}(s)&= \mathbb{E}[e^{-s\mathcal{I}_{\text{in}\mid \mathtt{h}>\tau , \mathtt{n}}}] = \frac{\exp\{- \mathtt{n} \tau s \rho_{\text{L}}\}}{(1+s\rho_{\text{L}})^\mathtt{n}}.
\normalsize
\end{align}
The probability mass function of the number of neighbors  $\mathcal{N}$ which is found in \cite{6576413} as: 
\begin{align}\label{pdf_users11}
\mathbb{P}\{\mathcal{N} = \mathtt{n}\} \approx \frac{\Gamma(\mathtt{n}+c)}{\Gamma(\mathtt{n}+1)\Gamma(c)} \frac{\left(\frac{\delta\; \mu}{n_Z}\right)^\mathtt{n} (\lambda c)^c}{\left(\frac{\delta\; \mu}{n_Z}+\lambda c\right)^{\mathtt{n}+c}}.
\end{align}
Considering that there is only Inter-cell interference when the number of neighbors in the cell is 0, and both of inter-cell and intra-cell interference otherwise we can write equation~(\ref{SINR_RA_2}) as \eqref{eq:AppA_2}. 
	\begin{align}\label{eq:AppA_2}
{P_{\text{{\rm RA}}}}&=  \exp\left\{\!- \frac{\sigma^2 \theta_{\text{{\rm RA}}}}{\rho_{\text{L}}}+\tau \right\} \mathscr{L}_{\mathcal{I}_{\text{out}\mid  \mathtt{h}>\tau}} \left(\frac{\theta_{\text{{\rm RA}}}}{\rho_{\text{L}}} \right) \notag \\ & \quad \times \left[ \mathbb{P} \left\{\mathcal{N}=0 \right\}+\!{\sum\limits_{\mathtt{n}=1}^{\infty}\mathbb{P} \left\{\mathcal{N}=\mathtt{n} \right\} \times \mathscr{L}_{\mathcal{I}_{\text{in}\mid  \mathtt{h}>\tau , \mathtt{n}}}\left(\frac{\theta_{\text{{\rm RA}}}}{\rho_{\text{L}}} \right) } \right].
\end{align}
To take into account the boundaries in \eqref{SINR_RA_2}, we use Gil-Pelaez theorem \cite{Gill}. Therefore, the CDF of the aggregated interference $F_{\mathcal{I}}\left(x \right)$ can be calculated by: 
	\small
\begin{align}\label{eq:CDF_gil}
F_{\mathcal{I}}\left(x \right)=  &\frac{1}{2}-\frac{1}{\pi} \int\limits_{0}^{\infty} \frac{1}{t} \text{Im} \left\{ \exp\left\{ -j \;t \;x \right\} \mathscr{L}_{\mathcal{I}}\left(-j\;t\right)\right\} dt,
\end{align}
\normalsize
where $\mathscr{L}_{\mathcal{I}}$ is the Laplace Transform of the aggregated interference which has the form of: 
 	\small
 \begin{align}\label{eq:lap_I}
\mathscr{L}_{\mathcal{I}}(s)=&  \mathscr{L}_{\mathcal{I}_{\text{out}\mid  \mathtt{h}>\tau}} \left(s\right) \notag \\ & \quad \times \left[ \mathbb{P} \left\{\mathcal{N}=0 \right\}+\!{\sum\limits_{\mathtt{n}=1}^{\infty}\mathbb{P} \left\{\mathcal{N}=\mathtt{n} \right\} \times \mathscr{L}_{\mathcal{I}_{\text{in}\mid \mathtt{h}>\tau , \mathtt{n}}}\left(s \right) } \right].
\end{align}
 \normalsize
After Applying the total probability theorem \eqref{eq:RA_success2} is obtained.

%\balance
\bibliographystyle{IEEEtran}
\bibliography{IEEEabrv,StringDefinitions,refrences}
\vfill

\end{document}